%% file: main.tex
\documentclass{article}

\input{preamble.tex}
\hypersetup{
        colorlinks,
        linkcolor={green!75!black},%
        citecolor={blue!75!black},%
        urlcolor={blue!85!black}
}
\usepackage{lineno}

\usepackage{todonotes}
\usepackage{comment}

\usepackage[english]{babel}

\usepackage[letterpaper,top=2cm,bottom=2cm,left=3cm,right=3cm,marginparwidth=1.75cm]{geometry}

\usepackage[shortlabels]{enumitem}

\title{Online List Labeling in Linear Space with Near-Logarithmic Writes}
\title{Online List Labeling with Near-Logarithmic Writes}

\author{Martin P. Seybold\thanks{\email{martin.seybold@univie.ac.at}}}

\date{~}
\begin{document}
\maketitle

\begin{abstract}
\input{0-abstract} 
\end{abstract}

\section{Introduction} 

\input{1-introduction}

\section{Using the Skip-List for Maintaining an Allocation} \label{sec:algo}

\input{2-algorithm}

\section{Analysis of the Run-Length and the Weight-Balance}\label{sec:analysis}
\input{3-analysis}

\paragraph*{Acknowlegments}
The author wants to thank the anonymous reviewers for pointing out an error in a draft on an early form of the algorithm that resulted into an overflow on the workload from~\cite{DietzZ90-log2N-lower-bound,zhang1993density}, cf. new Lemma~\ref{lem:no-overflow} and Section~\ref{sec:adaptive}.
Special thanks to Maximilian Vötsch for many discussions and collaborating on the earlier version of this paper.

This project has received funding from the European Research Council (ERC)  under the European Union's Horizon 2020
\marginpar{\includegraphics[height=22pt]{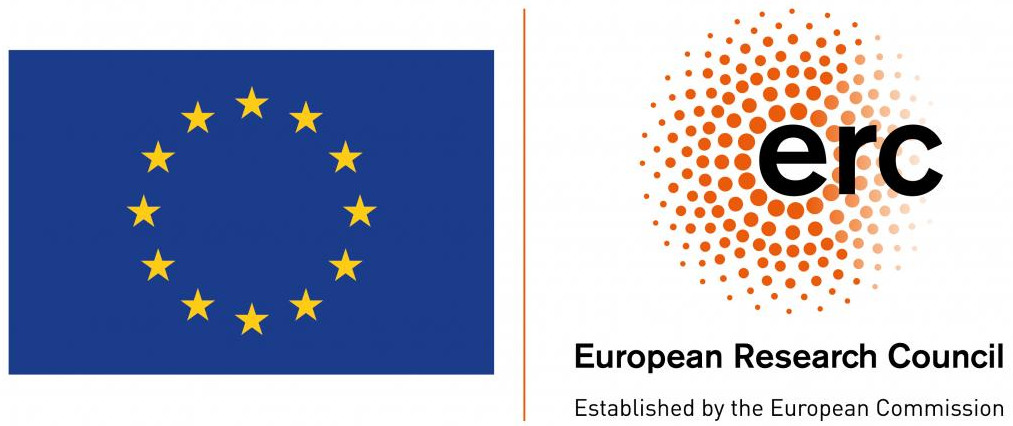}}
research and innovation programme (Grant agreement No. 101019564) and the Austrian Science Fund (FWF) project Z~422-N, project I~5982-N, and project P~33775-N, with additional funding from the \textit{netidee SCIENCE Stiftung}, 2020--2024.

\bibliographystyle{beta}
\bibliography{bibliography}

\end{document}

%% file: preamble.tex
\usepackage{amsmath,amssymb,amsfonts,amsthm}
\usepackage{subcaption}
\usepackage{graphicx}
\usepackage{fullpage}
\usepackage[backref=page]{hyperref}
\usepackage{color}
\usepackage{wrapfig}
\usepackage{tikz}
\usetikzlibrary{decorations.pathreplacing}
\usepackage{setspace}
\usepackage{algorithm}
\usepackage[noend]{algpseudocode}
\usepackage[framemethod=tikz]{mdframed}
\usepackage{xspace}
\usepackage{pgfplots}
\usepackage{framed}
\usepackage{thmtools}
\usepackage{thm-restate}
\usepackage{booktabs}
\pgfplotsset{compat=1.5}

\newtheorem{theorem}{Theorem}[section]

\newtheorem{lemma}[theorem]{Lemma}

\newtheorem{remark}[theorem]{Remark}

\newenvironment{proofof}[1]{\begin{trivlist} \item {\bf Proof
#1:~~}}
  {\qed\end{trivlist}}

\newcommand{\namedref}[2]{\hyperref[#2]{#1~\ref*{#2}}}

\newcommand{\eqnlab}[1]{\label{eq:#1}}

\newcommand{\eps}{\varepsilon}

\newcommand{\mdef}[1]{{\ensuremath{#1}}\xspace}  %

\DeclareMathOperator*{\poly}{poly}

\newcommand{\abs}[1]{\mdef{\left|#1\right|}}         %
\newcommand{\E}[2][]{\mdef{\underset{#1}{\mathbb{E}}\left[#2\right]}} %

\newcommand{\ignore}[1]{}

\newif\ifnotes\notestrue %
\newif\ifnotes\notesfalse
\ifnotes
\newcommand{\mv}[1]{\textcolor{purple}{{\bf (Max:} {#1}{\bf ) }} \marginpar{\tiny\bf
             \begin{minipage}[t]{0.5in}
               \raggedright :3 
            \end{minipage}}}            							
\else
\newcommand{\mv}[1]{}
\fi
\ifnotes
\newcommand{\mps}[1]{\textcolor{red}{{\bf (MPS:} {#1}{\bf ) }} \marginpar{\tiny\bf
             \begin{minipage}[t]{0.5in}
               \raggedright :3 
            \end{minipage}}}            							
\else
\newcommand{\mps}[1]{}
\fi

\renewcommand{\hat}{\widehat}
\usepackage{mathabx}
\renewcommand{\bar}{\widebar}

\makeatletter
\renewcommand*{\@fnsymbol}[1]{\textcolor{mahogany}{\ensuremath{\ifcase#1\or *\or \dagger\or \ddagger\or
 \mathsection\or \triangledown\or \mathparagraph\or \|\or **\or \dagger\dagger
   \or \ddagger\ddagger \else\@ctrerr\fi}}}
\makeatother

\providecommand{\email}[1]{\href{mailto:#1}{\nolinkurl{#1}\xspace}}

\definecolor{bleudefrance}{rgb}{0.19, 0.55, 0.91}
\definecolor{mahogany}{rgb}{0.75, 0.25, 0.0}

\hypersetup{
     colorlinks   = true,
     citecolor    = mahogany,
     urlcolor     = mahogany,
	 linkcolor	  = bleudefrance
}

%% file: 0-abstract.tex
In the Online List Labeling problem, a set of $n \leq N$ elements from a totally ordered universe must be stored in sorted order in an array with $m=N+\lceil\varepsilon N \rceil$ slots, where $\varepsilon \in (0,1]$ is constant, while an adversary chooses elements that must be inserted and deleted from the set.

We devise a skip-list based algorithm for maintaining order against an oblivious adversary and show that the expected amortized number of writes is $O(\varepsilon^{-1}\log (n) \poly(\log \log n))$ per update.

%% file: 1-introduction.tex
The \textsc{Online List Labeling} (OLL) problem, also known as the `file maintenance problem' as defined above, is a basic algorithmic task, rediscovered under various names, with numerous applications and a rich literature of solution strategies~\cite{dump-AnderssonL90,dump-BabkaBCKS19,dump-BenderDF05,dump-BenderFGKM17,dump-BenderH07,dump-BrodalDFILM10,dump-BrodalFJ02,dump-DevannyFGK17,dump-Raman99,dump-Willard82,dump-Willard92,saks2018online}.
In the OLL problem, we must store a set $X$ of $n\leq N$ keys from a totally ordered universe inside an array of size $m = N + \left\lceil \eps N \right\rceil$, while an adversary performs insertions and deletions on the set $X$.
The goal is to minimize the number of write operations in the array needed to maintain an ordered allocation, while processing the insertions and deletions into $X$ chosen by the adversary. 
Invariant of the problem is that $\abs{X} \in [0, N]$. %

In the polynomial regime, i.e., when $m = \poly{N}$, various solutions are known for obtaining an $O(\log n)$ write bound (cf. Table~1 in the survey~\cite{saks2018online}).
See also~\cite[Thm.~1]{BenderCDFZ02} and \cite[Thm.~5.3]{BlellochG07} based on Scapegoat~Trees~\cite[Thm.~4]{Andersson89} or Treaps~\cite[Thm.~4.7(ii)]{seidel1996randomized}.

Determining the complexity of the asymptotic number of write-operations needed for Online List Labeling in the linear regime, that is when $m=\Theta(N)$ is, however, still an important open problem\footnote{See~\cite[Section~8]{ItaiKR81} and the recent breakthrough work~\cite[Section~1]{bender2022online}.}.

Itai et al.~\cite{ItaiKR81} gave a deterministic algorithm with $O(\eps^{-1} \log^2 n)$ amortized write-cost, which matches the known $\Omega(\log^2 n)$ lower-bound of~\cite{BulanekKS12-log2n-lb-det-conf,BulanekKS12-log2n-lb-det-journal} for deterministic algorithms.
Dietz and Zhang~\cite{DietzZ90-log2N-lower-bound} showed that all algorithms that are `smooth'\footnote{A labeling algorithm is called smooth, if \emph{``the list items relabeled before each insertion form a contiguous neighborhood of the list position specified for the new item, and the new labels are as widely and equally spaced as possible.''} (see~\cite[p627]{DietzSZ04-journal-wc} and `gaps' in~\cite[p175]{DietzZ90-log2N-lower-bound}).} must have $\Omega(\log^2 n)$ write-cost, which applies even for randomized algorithms that (merely maintain fractional allocations and) are given the entire update sequence in advance (offline).
This lower bound is obtained on an insertion-only sequence where the adversary inserts the keys always at the front, i.e. in descending key order.
(See Chapter~7 in \cite{zhang1993density}.)

The recent breakthrough work of Bender, Conway, Farach-Colton, Komlós, Kuszmaul, and Wein \cite{bender2022online} showed how to obtain an expected $O(\eps^{-1}\log^{3/2} n)$ write-cost for updates.
Their solution is history-independent and uses a random shift and a certain top-down sampling condition that results in a weight-balanced Randomized Search Tree.
The authors also show that, for constant $\eps \leq 0.01$, every history-independent solution has $\Omega(\eps^{-1} \log^{3/2} n)$ expected write-cost.
This remarkable lower bound is obtained on an update sequence where the adversary alternately deletes the largest element from the set $X$ and inserts a random element with a uniform rank. %
(See Theorem~18, and Lemma~23, in~\cite{bender2022full}.)

The strongest general lower bound for the write-cost is $\Omega(\log n)$, due to the work of Bulánek, Koucký, and Saks~\cite{BulanekKS13-logN-lower-bound}.
We give the following remark.

\begin{remark}
One of the main obstacles in using the well-known randomized search trees, e.g.~\cite{seidel1996randomized}, for solving the problem is that they do not naturally provide a `good' weight-balance: 
The root of a randomized search tree on $n$ keys has expected parent-child weight-balance
$\geq \frac{1}{n} \sum_{i=1}^{\lfloor n/2\rfloor } \tfrac{n}{i} = \Theta( \log n)$.
\end{remark}

\subsection{Contribution, Paper Outline, and Open Problems}
We devise an algorithm that solves the Online List Labeling problem with expected amortized write-cost $O(\eps^{-1} \log (n)\log \log n)$ per update, assuming $\eps \in (0,1]$ is constant.
The proposed algorithm improves on the recent breakthrough result of~\cite{bender2022online} for the linear-space regime.

In Section~\ref{sec:algo}, we introduce a randomized algorithm that uses the interval tree, induced by a skip-list, to define a smooth allocation mechanism, which utilizes basic properties of skip-lists that we show in Section~\ref{sec:analysis}.
To ensure that we can process all updates using this approach, we introduce a `proactive reallocation' mechanism that depends on the history of insertions and deletions performed by the adversary (cf. Lemma~\ref{lem:no-overflow}).
Section~\ref{sec:smooth-gamma} shows that a parameter of $\Theta(1/\log n)$ is sufficiently small for proactive reallocations that use the smooth allocation mechanism, leading to a write-bound that matches the known lower bound for the class of smooth algorithms.
Section~\ref{sec:adaptive} then introduces our non-smooth allocation mechanism that allows to raise the threshold parameter of the proactive mechanism to a constant, at expense of repeating the non-smooth allocation $O(\log \log n)$ times.

The proposed algorithm and analysis are with respect to $\eps$ being a constant in $(0,1]$.
An open problem for future work is the technical, nevertheless important, extension to non-constant $\eps$, for example $\eps = \Theta(1/\log n)$.
Another open problem is whether obtaining algorithms with an $\widetilde{O}(\eps^{-1} + \log n)$ write-bound is possible for the problem. 
One of our lemmas, Lemma~\ref{lem:parent}, seems to allow a sharper upper bound than $O(\eps^{-1}\log n)$. 
However, the $1/\eps$ factor in our final bound is very closely related to the proactive mechanism of our algorithm, which is why additional algorithmic ideas seem necessary in attempts to improve the upper bound.

%% file: 2-algorithm.tex
Let $\eps \in (0,1]$ be constant.
By the standard periodic-rebuilding argument, it suffices to give an algorithm for the case that the number of keys $n$ in the dynamic set $X$ is within 

\begin{equation} 
(1-\eps/4)\bar n   \quad \leq \quad   n   \quad \leq   \quad (1+\eps/4)\bar n 
\end{equation}

 of some $\bar n \leq N$, since rewriting an entirely new allocation in the array whenever the number of updates\footnote{Counting both, insertions and deletions, to $X$.} $\delta$ to $X$ reaches a constant fraction of the set's size at the time of the last rebuild, 
i.e. $\delta = \lfloor \eps \bar n/4\rfloor$, has amortized
$O(\bar n + \delta)/\delta=O(\eps^{-1})$ 
write-cost per update.
By using only the first 
$m':=\min\{m, \lceil (1+\eps) (1+\eps/4)\bar n\rceil\}$ array slots, an algorithm may further assume that the density remains high
$\frac{n}{m'}>1-2\eps$ until the next periodic-rebuild.
That is, the density remains in the range

\begin{equation} \label{eq:invariant-density}
    1-2\eps
      < \frac{1-\eps/4}{(1+\eps)(1+\eps/4)+1/\bar n} 
\leq  
    \frac{n}{m'}
\leq
    \frac{1+\eps/4}{(1+\eps)(1+\eps/4)}
    =1-\frac{\eps}{1+\eps} \leq 1-\frac \eps 2 ~.
\end{equation}

This, seemingly arbitrary, focus on a dense parameter constellation 
will be crucial for a simple analysis of the write-efficiency (\autoref{thm:main-result}) of our `proactive' algorithm of this section.

To simplify the presentation, we assume that $X$ is a set of $n$ points from $\mathbb{R}$. 
This assumption poses no restriction in what follows, since we will merely require that the open interval $(x,x')$ of two distinct keys $x,x' \in X$ is well-defined on the totally ordered universe.

Our algorithm uses the interval tree that is induced by a skip-list~\cite{Pugh90} to maintain a mapping of the keys in the skip-list to positions in the array.
(See Chapter~1.4 in~\cite{mul-book} for the description of skip-lists on which  our notation is based on.)
A skip-list for the keys $X$ consists of $r \geq 1$ levels that store a `gradation' of the elements, i.e. 
\[
    X=X_1 ~\supseteq~ X_2 ~\supseteq~ \dots ~\supseteq~ X_{r-1} ~\supseteq~ X_r = \emptyset~,
\]

where the keys $X_{\ell+1}$ are obtained, bottom-up, from those keys in $X_\ell$ by independent tosses of a fair coin.
The keys in $X_\ell \setminus X_{\ell+1}$ are called {\bf \em keys on level $\ell$}. 
Each $X_\ell$ induces a partition of the real line $\mathbb{R} \setminus X_\ell$ into \emph{open} intervals, which we call the {\bf \em intervals of level $\ell$}.
The interval $(-\infty,+\infty)$ on level $r$ is called the root, and the skip-list {\bf \em size} is the number of its open intervals 
$\left(2r + 2\sum_{x\in X} level(x)\right)/2$.

For a level $\ell$ interval $I$ we call the level $\ell+1$ interval $U \supseteq I$ the {\bf \em parent} of $I$.
The level $\ell$ intervals contained in an interval $U$ of level $\ell+1$ are called the children of the interval $U$.
Note that the parent has only one child, if both intervals have equal boundaries ($U=I$).
Let $X(U) := X \cap U$ denote the set of keys that are contained in the interval $U$, and the {\bf \em weight} of a level $\ell$ interval $U$ is the number of intervals that form its subtree (including the subtree root $U$), i.e. 
\begin{align*}
    w(U)    =   \ell+\sum_{x\in X(U)} level(x)  ~~>~~ |X(U)|~.
\end{align*}
Thus, weight (strictly) decreases on any root-to-leaf path, and remains
 $w(U) \geq 1+|X(U)|$ and $w(U)\leq \ell+|X(U)|(\ell-1)$.
Since $\Pr[level(x)=i]=1/2^i$ and $\E{level(x)}=2$, we have for any level $\ell$ interval $U$ that
$ \E{w(U)}\leq \ell + 2|X(U)|$.

For an interval $U$ that is the parent of the intervals $(I_1,\ldots,I_d)$, we call the keys 
$
    sep(U) := X(U) \setminus (X(I_1) \cup \ldots \cup X(I_d))
$ 
the {\bf \em separators} of interval $U$. 
Note that there are exactly $d-1$ keys in $sep(U)$ since the $d\geq1$ children are open intervals, i.e. key intervals of the form $(x,x')$, and that all keys are on the same level (one beneath the level of $U$).

For example, an interval $U$ that is the parent of $(I_1,\ldots,I_d)$ has $d-1$ separator keys and weight $w(U) = 1+w(I_1)+\ldots +w(I_d)$. 
The weight of intervals on the lowest level $X_1$ is one, as they contain no keys from $X$.
See Figure~\ref{fig:skip-list} (top).

\begin{figure}
    \centering
    \includegraphics[width=\textwidth]{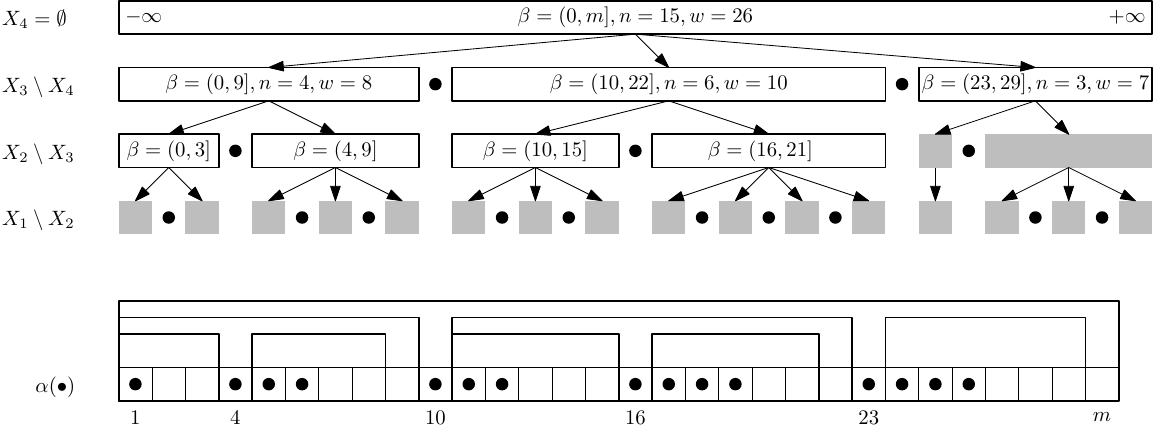}
\caption{
    Interval tree of a skip-list (top) and allocation of keys (bottom).
    Allocated intervals and their respective index-ranges are shown as solid boxes, non-allocated intervals in gray.
    The keys on level $\ell$ are $X_\ell \setminus X_{\ell+1}$ and the (open) intervals of level $\ell$ are $\mathbb{R} \setminus X_\ell$.
    The example shows an initial allocation with five runs.
}
    \label{fig:skip-list}
\end{figure}

\subsection{Maintaining Allocations by Reallocating Certain Subtrees}
To simplify presentation, we first discuss how skip-lists can be used to maintain a fractional allocation 
$\alpha : X \to [1,m]$, that is all $x\neq x'$ from $X$ have $\alpha(x') \notin (\alpha(x)-1,\alpha(x)]$.
The extension to integral allocations $\hat \alpha: X \to \{1,\ldots,m\}$, are due to a simple rounding scheme that uses the ceiling function and always places elements leftmost.

Given the interval tree of a skip-list, we assign an index range as {\bf \em budget} $\beta(I) \subseteq (0,m]$ to the intervals in a top-down fashion.
The root-interval $U$ has the budget $\beta(U):=(0,m]$ at all times, where $m$ is the total number of array slots.
Informally, our smooth allocation splits the slack, i.e. the unused budget $|\beta(U)|-|X(U)|$, among the children $(I_1,\ldots,I_d)$ of interval $U$ proportional to their weight~$w(I_i)$ and the $i$-th separator key in $sep(U)$ is stored accordingly  within the index-range $\beta(U)$ of the array, which determines our {\bf \em allocation} $\alpha : X \to [1,m]$.

Formally, for an interval $U$ with budget $\beta(U)=(a,b]$ let

\begin{align}
    &\text{``Slack'':}&~ 
    \Delta(U) := |\beta(U)|-|X(U)| = b-a -|X(U)|\\
    &\text{``Relative Slack'':}&~     
    \eps_U    :=   \frac{\Delta(U)}{w(U)} ~.
\end{align}

If $U$ has only one child, i.e. $sep(U)=\emptyset$, we reserve one slot for $U$ and assign the entire remaining budget $\beta(I)=(a,b-1]$ to this child. %
Otherwise, the $i$-th separator key $x_i \in sep(U)$ is assigned to 

\begin{align*}      \label{eq:separator-location}
    \alpha(x_1) := a 
        + &\left( |X(I_1)| + \frac{\Delta(U)-1}{w(U)-1} w(I_1)\right) + 1
\\
    \alpha(x_{i+1}) := \alpha(x_i) 
        + &\left(|X(I_i)| + \frac{\Delta(U)-1}{w(U)-1} w(I_{i})\right) + 1~.
\end{align*}

Let $a < s_1 < \ldots < s_{d-1} < b$ be the positions assigned to the keys of $sep(U)$.
We use those positions to partition the range $\beta(U)$ for the children so that all have equal relative slack $\eps_{I_1}=\ldots=\eps_{I_d}$, i.e.
\begin{align}
    \beta(I_1)&:=(a,s_1-1]\\
    \beta(I_i)&:=(s_{i-1},s_i-1]\\
    \beta(I_d)&:=(s_{d-1},b-1]   ~.
\end{align}

For example, this partition leaves \emph{one slot} at the rightmost end of the budget in $\beta(U)$ unused, and the total slack 
$
\frac{\Delta(U)-1}{w(U)-1}\sum_{i=1}^d w(I_i)=\frac{\Delta(U)-1}{w(U)-1}(w(U)-1) = \Delta(U)-1
$ is conserved.

Our top-down allocation of budget and separator positions proceeds recursively until an interval $U$ is reached where one of its children $\{I_i\}$ would receive a budget with less than one slot of slack
\begin{align}                           \label{eq:tdalloc-termination}
&\text{``Termination'':}&~   \Delta(I_i) &< 1 ~.%
\end{align}
If a child $I_i$ of $U$ satisfies the Termination criteria~\eqref{eq:tdalloc-termination}, we assign all keys $x \in X(U)$ as one {\bf \em consecutive run} of sorted keys to the index-range $\beta(U)$.
We write the run at an arbitrary position in $\beta(U)$, say leftmost.

Note that the last slot of $\beta(U)$ remains vacant, that is the Termination Criterion of our top-down approach yields that 
{\bf \em allocated intervals} have at least one unit of slack 

\begin{align}                           \eqnlab{invariant-fixed}
    &\text{``Allocation-Invariant'':}&~
    |\beta(U)|-\bar n(U)~=~\bar \Delta(U) \quad &\geq 1~, %
\end{align}

where $\bar n(U)$ denotes the interval's key-count $|X(U)|$ at the time of its allocation (of a budget).
We call all remaining intervals of the tree that are not assigned a budget by the top-down allocation {\bf \em {non-allocated} intervals}\footnote{
For example, every allocated interval $U$ has that either (i) all children $I$ of $U$ are allocated intervals, or (ii) all descendants of $U$ are non-allocated intervals. 
}.
This concludes the definition of our top-down allocation procedure, which we extend to our dynamic, smooth allocation next.

The `main idea' is to execute the top-down allocation repeatedly on certain subtree roots of the interval tree, which results in a {\bf \em reallocation} of the keys of the subtree.
To this end, we identify the interval tree nodes with the nested partition of the array's index-range they are allocated to. 
See Figure~\ref{fig:skip-list} (bottom).
This allows us to keep track of the current weight $w(U)$ and key-count $n(U)$ of interval $U$, its weight $\bar w(U)$ and key-count $\bar n(U)$ when interval $U$ was allocated last, the weight of the updated keys $\delta(U):=\sum_x level(x)$ that happened in $U$ after $U$ was allocated last, and the slack $\bar\Delta(U)$ when $U$ was allocated last.

Now, the insertion of a new key $x$ on skip-list level $\ell$ splits one interval of level $\ell$ into two intervals (left and right of the new element $x$), whereas deletion of $x$ from $X$ merges two level $\ell$ intervals that have $x$ as common key in their interval boundaries.
For maintaining an allocation $\alpha$ during updates to $X$, we proceed as follows:
After inserting/deleting a level $\ell$ key $x$ in the interval tree, we reallocate the entire subtree of the parent interval $U=parent(I)$ of the lowest allocated interval $I \ni x$ according to the current key-counts and weights within the space budget $\beta(U)$ that was allocated to the interval $U$.
If $I$ is the root of the interval tree, we simply rebuild the entire allocation of all keys by calling \verb-AllocateTD-($U$) on the new root, which also redistributes all budget with respect to the current weight of the root.

This strategy alone might fail once too many insertions into an interval have occurred. 
Indeed, if the space budget $\beta(U)$ of the parent interval $U$ does not increase while performing repeated insertions, %
 we eventually run out of space for inserting new keys.
This is when the current key-count has increased to $|X(U)| = |\beta(U)|$, thus violating the pre-condition $|\beta(U)|-|X(U)|\geq 1$ that is required for executing our top-down allocation from $U$.

To avoid this pathological case, our algorithm `proactively' ensures that top-down reallocations remain \emph{always} possible, as follows.
As soon as the weight-change $\delta(I)$ that happened in some allocated {interval~$I$} reaches a $\gamma$-fraction of %
its budget-slack\footnote{An early draft of this paper erroneously stated proportional to the {\bf key-count} here, instead of {\bf slack}, which is indeed insufficient to prevent overflows.}, i.e. when
\begin{align} \label{eq:proactive-triggering}
    \text{``Proactive $\gamma$-Trigger''}:&~
    & \delta(I)/\gamma \quad\geq \quad 
    \bar \Delta(I)~,
\end{align} 
for some parameter $\gamma \in (0,1]$, 
we reallocate the space budget $\beta(U)$ of the parent interval $U=parent(I)$ according to the current weights in the entire subtree of $U$. 
The amortized write-cost of such a reallocation is $\leq w(U)/\delta(I)$, as only the keys in $X(U)$ are rewritten.
Note that the parent-child relations of intervals remain unchanged in proactive reallocations, merely their index regions in the array change.
The dependence between the parameter $\gamma$ and the write-bound for our smooth and adaptive allocation formulas is analyzed in Sections~\ref{sec:smooth-gamma} and~\ref{sec:adaptive}, respectively.

\begin{lemma}\label{lem:no-overflow}
Let $\gamma \in (0,\tfrac12]$.
After every update operation, every allocated interval $U$ satisfies the pre-condition $\Delta(U) \geq 1$ of having at least one free slot.

In particular, reallocation from any allocated interval remains possible.
Thus, the allocation maintained by the interval tree does not overflow.
\end{lemma}

\begin{proof}
    Recall that $\Delta(U)=|\beta(U)|-|X(U)|$ denotes the current slack and $\bar \Delta(U)=|\beta(U)|-\bar n(U)$ the slack at allocation time.
    Consider an arbitrary sequence of update requests $(x_1,x_2,\ldots)$, starting after some complete rebuild of the entire (initial) allocation at the root of the interval tree.

    The proof is by contradiction.
    If an allocated interval violates after the $t$-th update $x_t$ to the interval tree the condition, we can consider the first time $t^\dagger$ that such a violation occurs.
    That is, after the update $x_{t^\dagger-1}$ concluded all allocated intervals in tree $T_{t^\dagger-1}$ are still non-violating and interval $U$, that was an allocated interval in $T_{t^\dagger-1}$, is now violating the condition in tree $T_{t^\dagger}$.
    Let $0 \leq t^* < t^\dagger$ be the time when concluding update $x_{t^*}$ allocated $\beta(U)$ for $U$.
    The key-count $|X(U)|$ attained through the updates from $(x_{t^*+1},\ldots,x_{t^\dagger})$ is thus $n(U) \geq |\beta(U)|$ in $T_{t^\dagger}$.
    Subtracting $\bar n(U)$ from both sides, this means that the increase 
    $$
            n(U) - \bar n(U) 
        ~\geq~ |\beta(U)|-\bar n(U) 
        ~=~ \bar \Delta(U)
        \quad,
    $$ where $\bar \Delta(U)$ is the slack at allocation time $t^*$,
    and 
    $\delta(U)\geq n(U)-\bar n(U)$ due to $level(x_i)\geq 1$. 
    
    Since any single update $x_i$ within $U$ increases the key-count $n(U)$ by at most $1$, the counter $\delta(U)$ must have been larger than 
    $n(U)-\bar n(U)-1$ in $T_{t^\dagger-1}$, i.e. before request $x_{t^\dagger}$ arrives.
    Thus, on the one hand the allocated interval $U$ has in $T_{t^\dagger-1}$ a counter value 
    $$
    \delta(U) \geq \bar \Delta(U)-1
    ~.
    $$
    But after the last update $t \in (t^*,t^\dagger)$ that did increment $\delta(U)$, we have on the other hand also a counter value 
    $$
    \delta(U) < \gamma \bar \Delta(U)
    ~,
    $$
    as the proactive threshold for a reallocation from $parent(U)$ was not triggered.
    
    We distinguish the two cases $\bar \Delta(U) =1$ and $\bar \Delta(U) \geq 2$ to complete the proof.

    If $\bar \Delta \geq 2$ this yields the contradiction, since 
    $
    \gamma \bar \Delta \leq \bar \Delta -1
    \Longleftrightarrow
    \gamma \leq 1 - 1/ \bar \Delta
    $ due to $\gamma \leq \tfrac12$.

    If $\bar \Delta = 1$, both inequalities $\delta \geq 0$ and $\delta < \gamma$ are only satisfied for $\delta(U)=0$.
    Contradicting that $U$ was allocated at time $t^*<t$.
\end{proof}

It remains to analyze the write-cost of the proposed algorithm for the cases that the smooth-split of the slack is used in the allocation formula~\eqref{eq:separator-location}, or the adaptive-split of the slack (Section~\ref{sec:adaptive}).
Informally, the proposed proactive algorithm aims at rebalancing the sibling's  slack as soon as one child received a number of updates that exceeds a $\gamma$-fraction of its allocated slack, causing its parent's budget being again split roughly equal among those siblings. 

Recall that the root's relative-slack remains in $[\eps/2,2\eps)$ between any two complete rebuilds at the root of the interval tree.
To motivate the following, let us briefly assume that \emph{any} allocated interval $U$ receives at any reallocation a relative-slack 
$\eps_U \in [\eps/2,2\eps)$ that remains within a constant fraction of the entire array.
(This is not necessarily true for every request sequence.)
Under this assumption, the allocated slack 
  $\bar \Delta(U)$ 
is always at least a constant fraction of the weight;
and we have for $\delta(U)$ from \autoref{eq:proactive-triggering} that:
\begin{align}
    \eps/2 \quad \leq\quad &\frac{\bar \Delta(U)}{\bar w(U)} \quad < \quad 2\eps \label{eq:slack-when-dense}\\
    \delta(U) \quad\leq\quad &\bar \Delta(U) \quad\leq\quad \delta(U)/\gamma~.
\end{align}

This will be the main observation for bounding the amortized write-cost $w(U)/\delta(I)$ of a proactive-reallocation in terms of $w(U)/w(I)$, since, assuming \autoref{eq:slack-when-dense}, it holds that 
\begin{align*}
    \frac{w(U)}{\delta(I)}
\leq 
    \frac{w(U)}{\gamma \bar\Delta(I)}
\leq 
    \frac{2 w(U)}{\gamma \eps \bar w(I)}
< 
    \frac{2 w(U)}{\gamma\eps \cdot w(I)/(1+2\eps)}
\quad \leq \quad \frac{4}{\gamma\eps} ~ w(U)/w(I)
~.
\end{align*} 
Note that the second inequality uses the lower bound $\bar\Delta(I)\geq\eps\bar w(I)/2$ and the third inequality uses the upper bound of \autoref{eq:slack-when-dense}  for 
$\bar w(I) \geq w(I)-\delta(I) \geq w(I)-\bar\Delta(I) > w(I)-2\eps \bar w(I)$.

It is clear that after every periodic rebuild at the root, the relative slack of every allocated interval is exactly that of the root, which is, by \autoref{eq:invariant-density}, in $[\eps/2,2\eps)$ throughout any \emph{phase} between two periodic rebuilds.
Let $T_t$ denote the interval tree after the $t$-th update request (of the phase) and smallest relative slack therein is denoted by 
    $
    \eta_t := \min \left\{ \eps(U)= \frac{|\beta(U)|-|X(U)|}{w(U)}~:~U \text{ allocated in $t$-th intervall tree } T_t \right\}
    $.

\subsection{Setting $\gamma=\Theta(1/\log n)$ Prevents Smooth-Allocations from `Drifting'}
\label{sec:smooth-gamma}
With the simplistic, smooth allocation formula discussed so far, it is not true that $\eta_t$ is within a constant factor of the relative slack of the root $\Omega(\eps)$ for all parameter values $\gamma \in (0,\tfrac 1 2 ]$.
Let us briefly discuss the adversary workload from \cite{zhang1993density,DietzZ90-log2N-lower-bound} that indeed causes a `drift to low slack', before we discuss the non-smooth fractional allocation formula from the next section that is designed to avoid drifts in the skip-list based approach also with larger $\gamma$ values.

In the phase, the adversary always inserts a new smallest key to $X$, i.e. $x_1 > x_2 > \ldots$ and so forth.
Though the skip-list is history independent, insertions are always close to key inserted last.
Consider the left spine of ancestor intervals 
$U_r \supsetneq U_{r-1} \supsetneq \ldots~$ 
of the interval tree.
If unlucky, all inserted keys are on, or close to, the lowest level. 
Then \emph{only} proactive-reallocations happen from the intervals of the left spine that are on the top-most levels (, which leave the parent/child relations of intervals unaffected).
Initially, all intervals have equal relative slack 
$\eps_r = \eps_{r-1}=\ldots~$ 
before the first insertion $x_1$ of the phase.
During the phase, the insertions decrease the relative slack of nodes, until the proactive threshold of a node $U_d$ is reached that in turn increases descendants relative slacks $\eps_{d+1},\eps_{d+2},\ldots$ all to the value of $\eps_d$.
The \emph{drift} of the relative-slacks $\{\eps_i\}$, away from $\eps_r$ of the root, can be bounded using the proactive threshold of \autoref{eq:proactive-triggering} as follows.

Consider $U_{r-1}$ below the root and let 
   $\bar \eps_{r-1}:=\bar\Delta(U_{r-1})/\bar w(U_{r-1})$ 
be its initial relative slack.
The smallest possible relative slack $\eps_{r-1}$, before $\delta_{r-1}$ triggers a reallocation from the root $U_r$, is 

\begin{align} 
    \eps_{r-1}
&=
    \frac{|\beta(U_{r-1})| - n(U_{r-1})}{w(U_{r-1})}
\geq 
    \frac{|\beta(U_{r-1})|-\bar n(U_{r-1})-\delta_{r-1}}{w(U_{r-1})}
\nonumber 
\\
&>
    \frac{\bar \eps_{r-1}\bar w(U_{r-1})-\gamma\bar\eps_{r-1}\bar w(U_{r-1})}{w(U_{r-1})}
\quad=\quad 
    \bar \eps_{r-1}(1-\gamma)\frac{\bar w(U_{r-1})}{w(U_{r-1})}
\label{eq:drift-bounded-delta}
\\
w(U_{r-1})
&\leq \bar w(U_{r-1})+\delta_{r-1} < \bar w(U_{r-1}) (1+ \gamma\bar \eps_{r-1})
\label{eq:drift-expected-weight-increase}
\\
\eqref{eq:drift-bounded-delta}\&\eqref{eq:drift-expected-weight-increase}\Longrightarrow\quad 
\eps_{r-1}
&\geq  
\bar \eps_{r-1} \frac{1-\gamma}{1+\gamma\bar\eps_{r-1}}
=
\bar \eps_{r-1} \left(1-\gamma\frac{1+\bar\eps_{r-1}}{1+\gamma\bar\eps_{r-1}} \right)
\geq 
\bar \eps_{r-1} (1-2\gamma)
\quad,
\label{eq:wc-drift}
\end{align}

where we use that
$
    \delta < \gamma\bar\Delta=\gamma \bar \eps \bar w
$
for \eqref{eq:drift-bounded-delta},
that the insertion of %
keys increases the weight by $\leq \delta$ for \eqref{eq:drift-expected-weight-increase},
and the last inequality uses that 
$
 \gamma\frac{1+\bar \eps}{1+\gamma\bar\eps}
<   \gamma(1+\bar\eps)
\leq 2\gamma
$.
That is, any smooth reallocation that descend from interval $U_{r-1}$ will allocate relative slacks of value at least 
$\eps_{r-1} \geq \bar \eps_r(1-2\gamma)$.
Now, the same argument applies to the relative slack $\eps_{r-2}$ of $U_{r-2}$, which shows that the expected relative slack is
    $\eps_{r-2} > \bar \eps_r(1-2\gamma)^2$.
Using Bernoulli's inequality\footnote{For $x\in(0,1]$ and integer $d\geq 1$, 
we have that $(1-x)^d\geq1-xd$ from the geometric series 
$
d\geq\sum_{i=0}^{d-1} (1-x)^i=\frac{1-(1-x)^d}{x}
$.}, 
we have that
intervals at distance $d$ from the root have a bound for their relative slack of the form
$$
    \eps_{r-d} 
\quad > \quad  
    \bar \eps_r ~ (1-2\gamma)^d
\quad \geq \quad 
    \bar \eps_r ~ (1-2\gamma d)
 ~.
$$
Since $d \leq r$ and $r=O(\log n)$ with high probability, one way to ensure that smooth-allocations have an expected drift $\E{\eta_t} = \Omega(\bar \eps_r)$ is due to simply setting the proactive parameter $\gamma = \Theta (1/\log n)$,
which will however only match the known $\Omega(\log^2 n)$ lower bound for smooth, fractional allocation algorithms~\cite{zhang1993density,DietzZ90-log2N-lower-bound}.

{
\color{blue}

}
Note that generating such a critical path for the proactive-reallocation mechanism, whose relative slack drifts far from that of the root, requires \emph{many insertions} that first degrade $U_{r-1}$, then degrade $U_{r-2}$, and so fort.
We use this observation to devise our non-smooth split-mechanism of the budget.

\subsection{Non-Smooth: Countering the Drift with Repeated, Adaptive Splits}
\label{sec:adaptive}

Recall that smooth allocations simply use the current relative slack $\eps_U$ to assign each child $\eps_U w(I_i)$ slack for its budget $\beta(I_i)$.
Thereby, the relative slack $\eps(I_i)$ when a child $I_i$ triggers proactive reallocation at $U$ is maximized, as any non-fair split must assign at least one child of $U$ less than $\eps(U)$ relative slack in order to assign more than $\eps(U)$ to a child.
If the adversary knows which interval received less than a fair share in a Proactive Reallocation, changing the requested insertions to reside in this interval would diminish the relative slack \emph{even faster} than with the smooth-split.
As noted in previous works, it is highly unintuitive how deviating from a smooth algorithm (fair splits) can eventually turn out advantageous for density control.

The following adaptive-split mechanism for a Proactive Rellaocation from $U$ computes for each of its $d$ subtrees a value $\bar \eps(I)$, assigns $\bar \eps(I) \bar w(I)$ slack to the budget $\beta(I)$ of child $I$, and allocates the entire subtree of $I$, using %
above's smooth-allocation.
Since it seems very difficult to compute one such split, that anticipates all future requests of the adversary during the budget-phase, we partition the budget-phase of $U$ into two halves, and each half into $O(\log \log n)$ rounds.
Whenever the $\delta(U)$ falls in the range of the next round, we recompute, and reexecute, the adaptive split of $U$.
We will exploit two observation, on above's analysis of the drift, to derive two cases where the drift of deep nodes is \emph{negligible}, and use those to define the adaptive-split and the rounds of a budget-phase.
Recall that the fan-out $d$ of interval $U$ remains fixed during its budget-phase, i.e. a change due to an update would trigger a reallocation from $parent(U)$.

\begin{lemma}[First\&Last are Drift Safe] \label{lem:adaptive-base-case}
For reallocation from $U$, regardless of it being triggered by a child $I$ or being due to scheduled periodic rebuilding between the budget allocation and renewal of $U$.
If the number $\delta(U)$ or the pending number 
$\delta^\dagger(U):=\gamma\bar\Delta(U)-\delta(U)$
until the budget $\beta(U)$ is renewed is `very small' in relation to its allocated slack, i.e. if 
$\min\{\delta(U),\delta^\dagger(U)\} \leq \bar \Delta(U)/(\kappa \log N)$ 
for a sufficiently large constant $\kappa\gg1$,
then using the fair-split for all such reallocations leads to a negligible drift factor in a weight-balanced tree, regardless of the parameter $\gamma$.
\end{lemma}

We thus use the fair-split for the first, and the last, round as base cases of our adaptive-split.

\begin{proof}
Let $\hat w_U$ be the subtree weight of $U$ when its current slack 
$\Delta(U)=\hat \eps_U \hat w_U$ 
is partitioned with the fair-split ($\bar \eps_I=\hat \eps_U$).
The relative-slack of child $I$ is at least

\begin{align}
    \eps_I
&=
    \frac{(\bar n(I) + \bar\eps_I\bar w_I)-n(I)}{w_I}
\geq 
    \frac{\bar \eps_I \bar w_I -\delta_I}{w_I}
\geq 
    \frac{\bar \eps_I \bar w_I -\delta_I}{\bar w_I+ \delta_I}
\\
&=
    \bar \eps_I\frac{\bar w_I -\delta_I/\bar\eps_I}{\bar w_I + \delta_I}
=
    \bar \eps_I\left( 1-\frac{\delta_I +\delta_I/\bar\eps_I}{\bar w_I+\delta_I}\right)
\geq \quad
    \bar \eps_I\left( 1-\frac{2\delta_I/\bar\eps_I}{\bar w_I}\right)
\quad
=
    \bar \eps_I\left( 1-\frac{2\delta_I}{\bar\eps_I \hat w_U}\cdot\frac{\hat w_U}{\bar w_I}\right)
\nonumber\\
&\geq
    \bar \eps_I\left( 1-\frac{2}{1-2\gamma}\cdot\frac{\hat w_U/\bar w_I}{\kappa \log N}\right)~,
    \label{eq:adaptive-base-cases}
\end{align}

where the inequality of \eqref{eq:adaptive-base-cases} uses that $\delta_I$ is upper bounded by the number of updates in $U$ that are possible in the first round or in the last round, both of which are $\gamma\bar \Delta_U/\kappa \log N$.
For the first round, we use that 
$\bar \eps_I = \bar \eps_U$ and $\hat w_U=\bar w_U$, thus the denominator is at least $\bar \Delta _U$.
For the last round, we use that $\bar \eps_I \geq \bar \eps_U (1-2\gamma)$ and 
$\hat w_U \geq \bar w_U$.
Thus, $\E{\hat w_U /\bar w_I}=O(1)$ implies $\E{\eps_I}\geq \bar \eps_I \left(1-\frac{O(1)}{\kappa \log N} \right)$.

Since $r=O(\log N)$ whp., $r$ such factors can reduce the lower bound for the relative-slack merely by a factor 
$
    \left(1-\frac{O(1)}{\kappa \log N}\right)^r 
\geq
    1-\frac{O(r)}{\kappa \log N} 
\geq 
    \tfrac 1 2
$ for a sufficiently large constant $\kappa > 1$.
\end{proof}

For the adaptive-split, we observe that every single update that increments $\delta(U)$ does not just decrease the slack $\Delta(U)$ by one, it also slightly reduces the uncertainty of the weights that the children of $U$ have in the near future.
To exploit this observation, we partition the phase's weight updates $[1,\delta^\dagger(U)]$ into rounds of geometrically decreasing widths around the center $\delta^\dagger(U)/2$, leading to $O(\log \log N)$ rounds for the first half of the range and $O(\log \log N)$ rounds for the latter half of the range, either having a width of $\delta^\dagger(U)/2^i$ for some integer $i=O(\log \log n)$.
Note that the number of children $d$ remains fixed during the budget-phase of  interval $U$.
Whenever the number $\delta(U)$ falls in the next round we repeat the adaptive split, instead of performing one split at the end once $\delta(U)$ reaches $\delta^\dagger(U)$ as the smooth Proactive Reallocations would do.
To effectively compensate the drift during one round, consider the partition of the allocated slack 
$\bar \Delta(U)=\sum_j B_j$ proportional to the widths $\{\delta^\dagger_j\}$ of the rounds, i.e. the rounds with index $j$ and $j'$ have $B_j/\delta^\dagger_j = B_{j'}/\delta^\dagger_{j'}\geq 1$.

Intuitively, our adaptive split uses the slack $B_j$ as `bonus' for the current round $j$ and splits the remaining slack $(\Delta(U)-1)-B_j$ with the fair split according to the current weights.
Formally, each of the $d$ children $\{I_i\}$ of $U$ receives in any reallocation during round $j$ the budget
\begin{align}\label{eq:adaptive-split-slack-new}
    &\text{``Adaptive-Split'':}&~ 
|\beta(I_i)| = n(I_i) + \frac{(\Delta(U)-B_j)-1}{w(U)-1} w(I_i)+\frac{B_j}{d}~.
\end{align}
During round $j$, we repeat this allocation whenever $\delta(U)$ increases by $\delta^\dagger_j/2d$.
This completes the description of our non-smooth allocation mechanism.

Note that the last step ensures that any child has $\delta(I_i)\leq B_j/2d$ throughout round $j$.

\begin{lemma}[Rounds are Drift Safe] \label{lem:adaptive-any-round-is-safe}
    Let round $j$ be one of the $O(\log \log N)$ rounds in the budget-phase of $U$ that is not the first round, and not the last round. 
    Then child $I$ of $U$ has a relative-slack 
    $\eps(I)$ that remains at least as large as the relative slack of $U$ at the end of round $j$, excluding $B_j$ slots.
\end{lemma}

\begin{proof}
    It suffices to consider insertions, as deletions increase slack.
    Let $
    \hat \eps
    $ be the coefficient of the middle term in~\autoref{eq:adaptive-split-slack-new} when the budget of child $I$ was computed last.
    Thus $\hat \eps$ is at least as large as the relative slack of $U$ at the end of round $j$, excluding $B_j$ slots.
    The relative-slack of child $I$ is throughout round $j$ at least
\begin{align}
    \eps_I
&\geq 
    \frac{ \hat \eps \bar w_I +\frac{B_j}{d} -\delta(I)}{w_I}
\geq 
    \frac{ \hat \eps \bar w_I +\frac{B_j}{2d} }{w_I}
\geq 
    \frac{ \hat \eps \bar w_I +\delta_I}{w_I}
    \geq \frac{ \hat \eps \bar w_I + \delta_I}{\bar w_I+\delta_I}
            \geq \frac{ \hat \eps \bar w_I}{\bar w_I} = \hat \eps
            \quad,
\end{align}
where the last inequality uses that $\frac{a+x}{b+x} \geq \frac{a}{b}$ iff. $ b \geq a$.
\end{proof}

%% file: 3-analysis.tex
We show the following theorem in this section.

\begin{theorem} \label{thm:main-result}
    Let $\eps \in (0,1]$ be constant.
    There is a randomized algorithm for the Online List Labeling problem whose amortized number of write operations per update is in expectation at most 
    $O(\eps^{-1} \log (n) \log \log (n))$.
    This bound holds against worst-case insertion and deletion choices of an oblivious adversary.
    That is an adversary that has no knowledge of the outcome of the random coin tosses that determine the underlying skip-list of our algorithm.
\end{theorem}

There are two kinds of reallocations in the array.
\begin{enumerate}
    \item Parent-Reallocations are those caused by insertions and deletions that change the interval tree.
    \item Ancestor-Reallocations are those caused by our proactive~strategy.
\end{enumerate}

Recall that a standard skip-list on $n$ keys %
has with high probability $r=O(\log n)$ levels and each interval has $O(1)$ children in expectation (conflict-list size), regardless of its level (e.g.~\cite[Chapter~1.4]{mul-book}). %
Moreover, the expected weight of a subtree modified by an insertion (or a deletion) is $O(\log n)$, since the affected key of the update has level $\ell$ with probability $1/2^\ell$, %
the expected number of keys to the left, or the right, that have level less than $\ell$ is $O(2^\ell)$, and there are with high probability $O(\log n)$ levels. (See~\cite[p23]{mul-book} for example.)

Since the adversary is oblivious, we may always assume that when an update to a key is performed, the level of the update is distributed geometrically, even for adversarial updates, as the adversary is not aware of the coin flips of our algorithm.

\begin{lemma}[Subtree Weight]\label{lem:parent}
    Let $\eta=\min_t \eta_t$ the smallest relative slack encountered thus far.
    The expected write-cost of a Parent-Reallocation is $O(\eta^{-1} \log n)$.
\end{lemma}
\begin{proof} %
    Recall that an allocated interval is one that is assigned a budget by the top-down algorithm, and a non-allocated interval is not, or no longer, assigned a budget on its own as its keys are stored in some consecutive run.
    Since any interval is either allocated or non-allocated, the expected write-cost is bounded by the sum of both expected weights.
    
    For allocated intervals, 
    let $\ell$ be the level number of the key insertion/deletion, and $\ell+1$ the level of the parent interval $U$ from where the allocation is issued.
    We have that the expected parent weight at level $\ell+1$ is at most
    $\E{~w(U)~|~\text{update has level }\ell ~} = O(2^\ell)$.
    Thus, the expected cost of writing the entire subtree of the parent $U$ is at most
    \[
        \sum_{\ell \leq r} \E{~w(U)~|~\text{update has level }\ell ~} \Pr[\text{ update has level }\ell~] 
    =   \sum_{\ell \leq r}   O(2^{\ell}) O(1/2^{\ell})
    =   O(\log n)~,
    \]
    since $r=O(\log n)$ is with high probability. 
    Since the subtree weight of allocated intervals is less than that of all intervals containing the update, the expected write-cost for allocated intervals is $O(\log n)$.

    For non-allocated intervals, the write-cost is the length of the consecutive run that contains the affected key.
    To see that the expected number of keys in a consecutive run is at most 
    $O(\eta^{-1}\log n)$, we consider the case that the top-down reallocation stops on some level $\ell+1$ interval $U$, as one of its child interval $I$ on level $\ell$ satisfies the termination criteria of~\eqref{eq:tdalloc-termination}, i.e. $\Delta(I) < 1$, where the designated slack $\Delta(I)\geq \eps(U) w(I)$ in smooth allocations.
    Thus, at allocation time $t$ of $U$ we have that unallocated child $I$ has 
    weight $w(I) \leq 1/\eps(U) \leq 1/\eta_t$.
    Recall that definition of $w(I)=\ell+\sum_{x \in X(U)} level(x)$, where $level(x)\geq 1$.
    Consequently, interval $I$ contains less than $|X(I)|<1/\eta_t$ keys.
    
    There are three possible cases for the light interval $I \subsetneq U$.
    It may have 
    form $(x,x')$ for two keys from $X$ both of which with level at least $\ell$ and one with level exactly $\ell$,
    form $(-\infty,x)$ for one key $x$ that has level $\ell$, or 
    form $(x,+\infty)$ for one key $x$ that has level $\ell$.

    Since the probability to observe in $<1/\eta$ independent trials of keys that one key has level $\geq\ell$ is at most $O(\eta^{-1}/2^{\ell})$, and the expected weight of a level $\ell+1$ interval is 
    $O(2^{\ell})$, we have for level $\ell$ that the expected length of a consecutive run is $O(\eta^{-1})$.
    Thus, the expected length of a consecutive run is at most
    \[
        \sum_{\ell \leq r} 
        \E{~w(U)~|~ I\text{ has level }\ell ~}\Pr[~I\text{ has level }\ell~] ~\leq~ \sum_{\ell \leq r} O(2^{\ell}) O(\eta^{-1}/2^{\ell}) =   O(\eta^{-1}\log n)~,
    \]
    since $r=O(\log n)$ with high probability.
\end{proof}

For the bound of the write-cost of Ancestor-Reallocation, it remains from Section~\ref{sec:smooth-gamma} to finally bound the expected weight-balance.
Next, we show that the expected parent-child weight-balance is constant, regardless of the child's level number. 

\begin{lemma}[Weight-Balance] \label{lem:parent-child-balance}
Let $I$ be an interval and $U$ its parent interval in the interval tree.
The expected weight ratio between $U$ and $I$ is constant. 
That is $\E{w(U)/w(I)}=O(1)$.
\end{lemma}

Since $w(U) > w(I)$, the random variables $w(U)$ and $w(I)$ are dependent.

\begin{proof}
Let $\ell+1$ be the level number of $U$ and $\ell=level(I)$ the level number of $I$.
Consider the partition 
\[
X(U) ~=~ X^- ~\cup~ X(I) ~\cup~ X^+
\]
of the keys in $U$ into, those keys $X(I)$ that are contained in $I$, those keys $X^-$ that are smaller than those keys in $I$, and those keys $X^+$ that are larger than those keys in $I$.
The proof shows a stronger property, namely that the key-count~$n^-=|X^-|$ and the key-count~$n^+=|X^+|$ have an upper bounds that do not depend on the key-count $n_I=|X(I)|$.
The lemma's bound on the weight-ratio follows from this, since the keys' levels in a skip-list are determined by independent coin tosses.

There are three possible cases for the form of interval $I$, these are 
$I=(x,+\infty)$ for some key $x$ that has level $\geq \ell$,
$I=(-\infty,x)$ for some key $x$ that has level $\geq \ell$, and
$I=(x,x')$ for two distinct keys where one has level $\ell$ and the other level $\geq \ell$.

In case $I=(x,+\infty)$, we have that $n^+=0$. 
Thus, it suffices to bound $\E{\frac{n^-+n_I}{n_I}}$. %
Consider an arbitrary point $q \in \mathbb{R} \setminus X$.
Walking to the left of $q$, we count keys, one by one, as the walk traverses them until we observe that a key has level $\geq\ell$.
Since the probability to stop counting is $\Theta(1/2^\ell)$, we have $\E{n_I}=\Omega(2^{\ell})$.
Moreover, counting again with $q$ being immediately left of $x$ until we observe a key that has level $\geq \ell+2$ yields a number of keys $n^- $ that has $\E{n^-}= O(2^{\ell})$.
The expectation bound for $n^-$ is an upper bound, as the argument assumes that there are an infinite number of keys to the left of $q$, which is not really the case as $X$ is finite.
Since these two key-intervals are disjoint, the random variables $n_I$ and $n^-$ are independent. %
Thus $\E{\frac{n^-+n_I}{n_I}}= \E{\frac{n^-}{n_I}} +1 = O(1)$, where the last inequality uses the aforementioned independence. 

The case $I=(-\infty,x)$ is symmetric so that we can argue analogously.

In case $I=(x,x')$, counting the keys in $I$ first shows that $\E{n_I}=\Omega(2^\ell)$ for this case.
Now, counting the keys $n^-$ to the left until we observe one key with level $\geq \ell + 2$ shows that $n^- =O(2^{\ell})$,
and counting the keys $n^+$ to the right until we observe one key with level $\geq \ell + 2$ shows that $n^+ =O(2^{\ell})$.
The expectation bounds for $n^-$ and $n^+$ are again upper bounds.
Since the three intervals don't overlap, the three outcomes are independent.
Thus, the expected ratio
$ 
        \E{\frac{n^- + n_I + n^+}{n_I}} 
=       \E{ \frac{n^-}{n_I}+ 1 +\frac{n^+}{n_I}}
=       O(1),
$
where 
the inequality uses the aforementioned independence.
\end{proof}

Since there are with high probability $O(\log n)$ levels in a skip-list, 
every update operation is contained with high probability in $O(\log n)$ intervals, each of which charges it $O(\log \log n)$ times during it's budget-phase.
As a consequence, the expected amortized write-cost of Ancestor-Reallocations is $O(\eps^{-1}\log (n) \log \log (n))$, as we have shown before that the amortized cost of proactive reallocation of any interval is $O(\eps^{-1} w(U)/w(I))$.

The only choice that the adversary has is the rank of the key $x$ that is inserted or deleted, but our analysis holds regardless of the rank of the requested update. 
For the adversary to obtain a larger write-cost for Parent-Reallocations, the adversary would need to choose a key $x$ for the update that is in an interval range of a subtree root with a weight that is larger than expected. 
For obtaining a larger amortized write-cost for Ancestor-Reallocations, the adversary would need to choose a key $x$ for updates in an interval range of two subtrees with a parent-child weight-balance larger than expected.
Either is only possible if the adversary knows the outcome of the coin tosses that define the underlying skip-list, and so the adversarial updates can, in expectation, cost no more than stated.